\documentclass[a4paper,11pt]{article}

\usepackage{fullpage}
\usepackage{amsfonts}
\usepackage{amssymb,amsmath}
\usepackage{latexsym}
\usepackage{graphicx}

\begin{document}

\bibliographystyle{plain}

\def\R{\mathbb{R}}
\def\N{\mathbb{N}}
\def\true{\mbox{\sc true}}
\def\false{\mbox{\sc false}}
\def\geq{\geqslant}
\def\leq{\leqslant}
\def\eps{\varepsilon}
\def\A{\mathcal{A}}
\def\C{\mathcal{C}}
\def\d{{\mathrm d}}
\def\un{{\mathbf{1}}}

\newtheorem{theorem}{Theorem}
\newtheorem{lemma}[theorem]{Lemma}
\newtheorem{proposition}[theorem]{Proposition}
\newtheorem{corollary}[theorem]{Corollary}
\newtheorem{question}{Question}

\def\boxit#1{\vbox{\hrule\hbox{\vrule\kern3pt
  \vbox{\kern3pt#1\kern3pt}\kern3pt\vrule}\hrule}}
\def\Box{\boxit{\null}}
\newcommand{\qed}{~$\Box$\medbreak}
\newenvironment{proof}{\noindent{\bf Proof: }}{\qed}

\title{A deterministic algorithm for fitting a step function to a weighted
point-set}

\author{Herv\'e Fournier\thanks{Univ Paris Diderot, Sorbonne Paris Cit\'e,
Institut de Math\'ematiques de Jussieu, UMR 7586 CNRS, F-75205 Paris, France.
Email: \texttt{fournier@math.univ-paris-diderot.fr}.
}
\and Antoine Vigneron\thanks{King Abdullah University of Science and
Technology (KAUST),
Geometric Modeling and Scientific Visualization 
Center, Thuwal 23955-6900. Saudi Arabia.
Email: \texttt{antoine.vigneron@kaust.edu.sa}.
}
}

\date{\today}

\maketitle

\begin{abstract}
Given a set of $n$ points in the plane, each point having a positive weight,
and an integer $k>0$, we present an optimal $O(n \log n)$-time deterministic algorithm
to compute a step function with $k$ steps that minimizes the maximum
weighted vertical distance to the input points.
It matches the expected time bound of the best known randomized algorithm
for this problem.
Our approach relies on Cole's improved parametric searching technique.
As a direct application, our result yields the first $O(n \log n)$-time
algorithm for computing a $k$-center of a set of $n$ points on the
real line.
\end{abstract}

\section{Problem formulation and previous works}

A function $f: \R \rightarrow \R$ is called
a $k$-\emph{step function} if there exists a real sequence
$a_1<\dots <a_{k-1}$ such that the restriction
of $f$ to each of the intervals $(-\infty,a_1)$, $[a_i,a_{i+1})$
and $[a_{k-1},+\infty)$ is a constant.
A weighted point in the plane is a triplet $p=(x,y,w) \in \R^3$
where $(x,y) \in \R^2$ represents the coordinates of $p$
and $w > 0$ is a weight associated with $p$.
We use $\d(p,f)$ to denote the weighted vertical distance 
$w \cdot |f(x)-y|$ between $p$ and $f$.
For a set of weighted points $P$,
we define the distance $\d(P,f)$ between $P$
and a step function $f$ as:
\[\d(P,f)=\max\{\d(p,f) \mid p \in P\}.\] 
Given the point-set $P$, 
our goal is to find a $k$-step function $f$ that minimizes $d(P,f)$.

This histogram construction problem is motivated by databases
applications, where one wants to find a compact representation
of the dataset that fits into main memory, so as to optimize
query processing~\cite{GuhaS07}. The unweighted version,
where $w_i=1$ for all $i$, has been studied extensively, until optimal
algorithms were found. (See our previous article~\cite{FournierV11}
and references therein.)

The weighted case was first 
considered by Guha and Shim~\cite{GuhaS07}, who gave an
$O(n\log n+k^2 \log^6 n)$-time algorithm. 
Lopez and Mayster~\cite{LopezM08} 
gave an $O(n^2)$-time algorithm, which is thus faster
for small values of $k$. 
Then Fournier and Vigneron~\cite{FournierV11} 
gave an $O(n \log^4 n)$ algorithm, which was
further improved to $O(\min(n \log^ 2 n, n\log n+ k^2 \log² \frac n k
\log n \log \log n))$ by Chen and Wang~\cite{ChenW09}.
Eventually, an optimal randomized $O(n \log n)$-time 
algorithm was obtained by Liu~\cite{Liu10}. 
In this note,
we present a deterministic counterpart to Liu's algorithm, which runs
in $O(n \log n)$ time. This time bound is optimal as the unweighted
case already requires $\Omega(n \log n)$ time~\cite{FournierV11}. 
Our approach combines ideas from
previous work on this problem~\cite{GuhaS07,KarrasMS07} with
the improved parametric searching technique by Cole~\cite{Cole87}.

Our result has a direct application to the $k$-center problem
on the real line: Given a set of $n$ points $r_1 < \dots < r_n \in \R$,
with weights $w_1,\dots, w_n$,
the goal is to find a set of $k$ centers $c_1,\dots,c_k \in \R$ 
that minimizes the maximum over $i$ of the  weighted distance 
$w_i \d(r_i,\{c_1,\ldots,c_k\})$.
Given such an instance of the weighted $k$-center
problem, we construct an instance of our step-function 
approximation problem where the input points are $p_i=(i,r_i)$  
for $i=1,\dots,n$,  keeping the same weights $w_1,\dots,w_n$. 
Then these two problems are equivalent: The $y$-coordinates
of the $k$ steps of an optimal step-function give an optimal
set of $k$ centers. So our algorithm also solves the weighted
$k$-center problem on the line in $O(n\log n)$ time, improving
on a recent result by Chen and Wang~\cite{ChenW11}.

\section{An optimal deterministic algorithm}

We consider an input set of weighted points
$P = \{(x_i,y_i,w_i) \mid 1 \leq i \leq n\}$, and an integer $k > 0$.
Let $\eps^*$ denote the optimal distance from $P$ to a $k$-step function,
that is, 
\[\eps^* = \min\{\d(P,f) \mid f\ \mbox{ is a } k\text{-step\ function}\}.\] 

Karras et al.~\cite{KarrasMS07} made the following observation:
\begin{lemma}\label{lem:decision}
Given a set of $n$ weighted points sorted with respect to their $x$-coordinate,
an integer $k > 0$ and a real $\eps > 0$,
one can decide in time $O(n)$ if $\eps < \eps^*$.
\end{lemma}
The above lemma is obtained by a greedy method, going through
the points from left to right and creating a new step
whenever necessary. More than $k$ steps are created along this process
if and only if $\eps < \eps^*$.
A consequence is that once $\eps^*$ is known, an optimal $k$-step
function can be built in linear time by running this algorithm
on $\eps=\eps^*$.

A second observation, made by Guha and Shim~\cite{GuhaS07}, is the
following.
The distance of a point $p=(x_i,y_i,w_i)$ to the constant function $c$
is equal to $\d(p,c)=w_i \cdot |c-y_i|$.
Hence, for a (non empty) subset $Q \subseteq P$ of the input points, 
the distance $\min \{\d(Q,f)\mid \mbox{$f$ is a constant function}\}$ 
between $Q$ and the closest constant function is given by the 
minimum $y$-coordinate of the
points in the region $U_Q$ defined as:
\[U_Q = \bigcap_{(x_i,y_i,w_i) \in Q} \{ (x,y) \mid y \geq w_i \cdot |x-y_i|\}.\]
In  other words, the distance between $Q$ and the closest 1-step function 
is the $y$-coordinate of the lowest vertex in the upper
envelope $U_Q$ of the lines with equation $y=\pm w_i(x-y_i)$ corresponding
to the points $(x_i,y_i,w_i) \in Q$.
(There is only one lowest vertex as the slopes $\pm w_i$ are nonzero.)

An immediate consequence 
is the following. For $i \in \{1,\ldots,n\}$, let $\ell_{2i-1}$ 
be the line defined by the equation $y = w_i(x-y_i)$, and
$\ell_{2i}$ the line defined by $y = -w_i(x-y_i)$.
Let $L = \{\ell_1,\ldots,\ell_{2n}\}$.
We denote by $\A (L)$ the arrangement of these lines. (See Figure~\ref{fig:envelope}.)
\begin{lemma}\label{lem:intersection}
The optimal distance $\eps^*$ from a set of weighted points
$P$ to a $k$-step function is the $y$-coordinate of a vertex of $\A (L)$.
\end{lemma}

\begin{figure}[ht]
\begin{center}
\scalebox{.5}{
\includegraphics{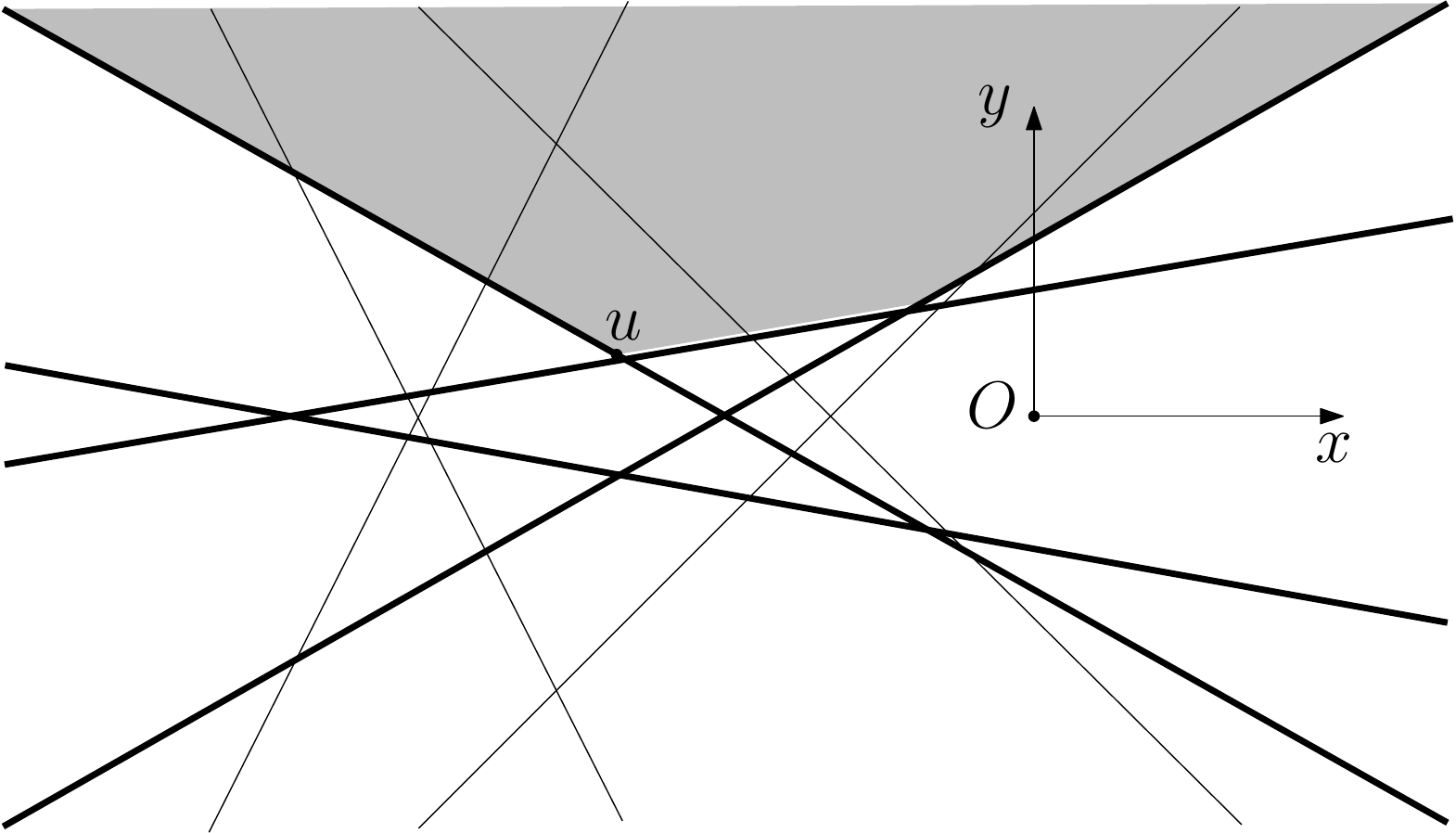}}
\end{center}
\caption{The arrangement $\A(L)$ and the upper envelope (shaded) of a subset $Q \subset L$
(bold).The $y$-coordinate of the lowest point $u$ in $U_Q$ gives the minimum distance
from $Q$ to a 1-step function.\label{fig:envelope}}
\end{figure}

%
%

The deterministic algorithm presented here
will be obtained by performing a search on the vertices
of $\A (L)$, calling the decision procedure of
Lemma~\ref{lem:decision} only $O(\log n)$ times,
and with an overall extra time $O(n \log n)$.
We achieve it by applying Cole's improved
parametric searching technique~\cite{Cole87}:

\begin{theorem}[Cole]\label{th:cole}
Consider the problem of sorting an array
$A[1,\ldots,n]$ of size $n$. Assume the following two conditions hold:
\begin{description}
\item{(i)} There is an $O(n)$ time algorithm to test if $A[i] \leq A[j]$.
\item{(ii)} There exists a linear order $\preceq$ on the set
$\{(i,j) \mid 1 \leq i < j \leq n\}$ such that
\[(i,j) \preceq (i',j')\ \Rightarrow\ \left( A[i] \leq A[j] 
\Rightarrow  A[i'] \leq A[j'] \right)\]
and such that we can decide 
if $(i,j) \preceq (i',j')$ in $O(1)$ time.
\end{description}
Then, the array $A$ can be sorted in $O(n \log n)$ time.
\end{theorem}

We briefly explain Cole's method.
Recall that a sorting network~\cite{Knuth1998sorting} 
for $n$ elements is a sequence $L_1,\ldots,L_d$,
each $L_i$ being a set of comparisons on disjoint inputs in $\{1,\ldots,n\}$.
On an input array $A[1,\ldots,n]$, the network operates as follows:
for each level $p$ from $1$ to $d$, the comparisons in $L_p$ are performed in
parallel, and the two elements of $A$ corresponding to each comparison 
are swapped if they appear in the wrong order.
If the sorting network is correct, the elements of $A$ are output in
sorted order after the last level.

The algorithm from Theorem~\ref{th:cole} works as follows.
First build a sorting network of depth $O(\log n)$ in deterministic
$O(n \log n)$ time~\cite{AjtaiKS83,Paterson90}.
During the course of the sorting algorithm,
each comparison in the network is marked with one of the
following labels: \emph{resolved}, \emph{active} or \emph{inactive}.
In the beginning, the comparisons at the first level $L_1$ of the network
are marked active, while all others are inactive.
The weight $1/4^p$ is assigned to each active node at level $p$.
Repeat the following until all nodes are resolved:
\begin{description}
\item{--} Compute the weighted median $(i_m,j_m)$ of all active comparisons
with respect to the order $\preceq$ defined in (ii);
\item{--} Decide if $A[i_m] \leq A[j_m]$ with the algorithm from (i).
This solves a weighted half of the active comparisons.
Swap the corresponding element of $A$ when in the wrong order, and
mark these comparisons as \emph{resolved}. Mark all inactive comparisons
having their two inputs resolved as \emph{active}.
\end{description}
It can be proved that at most $O(n)$ nodes are active at any step. Since the
weighted median can be computed in linear time, each step is performed in
$O(n)$ time. 
Moreover, it can be showed that the algorithm terminates in at most 
$O(\log n)$ steps.
So overall, this procedure sorts an array of size $n$ in $O(n \log n)$ time.

We are now ready to give our algorithm for fitting a step function
to a weighted point set:
\begin{theorem}\label{th:main}
Given a set $P$ of $n$ weighted points in the plane and
an integer $k>0$, a $k$-step function $f$ minimizing $\d(P,f)$
can be computed in $O(n \log n)$ deterministic time.
\end{theorem}

\begin{proof}
Let $\theta : \R \rightarrow \{0,1\}$ be the mapping defined by $\theta(\eps)=0$ if
$\eps < \eps^*$ and $\theta(\eps)=1$ otherwise.
First we sort the points of $P$ with respect to their $x$-coordinate.
Given $\eps$, this allows to compute $\theta(\eps)$ in time $O(n)$,
using the decision algorithm of Lemma~\ref{lem:decision}.

We denote by $\pi : \R^2 \rightarrow \R$ the projection onto $y$-coordinate axis.
Recall the definition of the lines $L = \{\ell_1,\ldots,\ell_{2n}\}$.
For $y \in \R$, let $\ell_i(y)$ be the unique $x \in \R$ such that $(x,y) \in \ell_i$;
it is well-defined since no line is parallel to the $x$-axis.
By Lemma~\ref{lem:intersection}, it holds that
\[\eps^* = \min \{ \pi(v) \mid v\ \text{vertex of}\ \A(L),\ \theta(\pi(v))=1\}.\]
Although $\eps^*$ is not known, we shall sort the set $\{ \ell_i(\eps^*) \mid 1 \leq i \leq 2n\}$
using Cole's method.

Note that $L$ could be of cardinality smaller than $2n$ if some lines are identical.
We discard identical lines and order them with respect to their order at $-\infty$.
That is, $L = \{\ell_1,\ldots,\ell_m\}$ and for all $0 < i < m$, it holds that
$\ell_i(y) < \ell_{i+1}(y)$ when $y \rightarrow -\infty$.
Let us check that conditions (i) and (ii) of Theorem~\ref{th:cole} hold.

Condition (i): Given $i < j$, we want to decide if $\ell_i(\eps^*) \leq \ell_j(\eps^*)$.
If lines $\ell_i$ and $\ell_j$ are parallel, the ordering on the lines ensures that
$\ell_i(y) < \ell_j(y)$ for all $y$ and in particular for $\eps^*$.
Otherwise, we compute $y_0=\pi(\ell_i \cap \ell_j)$ in time $O(1)$,
then decide if $y_0 < \eps^*$ in time $O(n)$ by Lemma~\ref{lem:decision}.
If $y_0 < \eps^*$, then $\ell_i(\eps^*) > \ell_j(\eps^*)$;
otherwise $\ell_i(\eps^*) \leq \ell_i(\eps^*)$.

Condition (ii): For $i<j$, let us define $\tilde\pi (\ell_i,\ell_j) = \pi(\ell_i \cap \ell_j)$
if lines $\ell_i$ and $\ell_j$ intersect, and $\tilde\pi (\ell_i,\ell_j) = +\infty$ otherwise.
(Or equivalently $\tilde\pi(\ell_i,\ell_j) = \sup \{y \in \R \mid \ell_i(y)<\ell_j(y)\}$.)
Let $\preceq$ be the order on the set $\{(i,j) \mid 1 \leq i < j \leq m\}$
defined by:
\[(i,j) \preceq (i',j')\ \text{if and only if}\ \tilde\pi(\ell_{i},\ell_{j}) \leq \tilde\pi(\ell_{i'},\ell_{j'}).\]
Assume $(i,j) \preceq (i',j')$ and $\ell_i(\eps^*) \leq \ell_j(\eps^*)$.  (See figure~\ref{fig:order}.)
\begin{figure}[ht]
\begin{center}
\scalebox{.5}{
\includegraphics{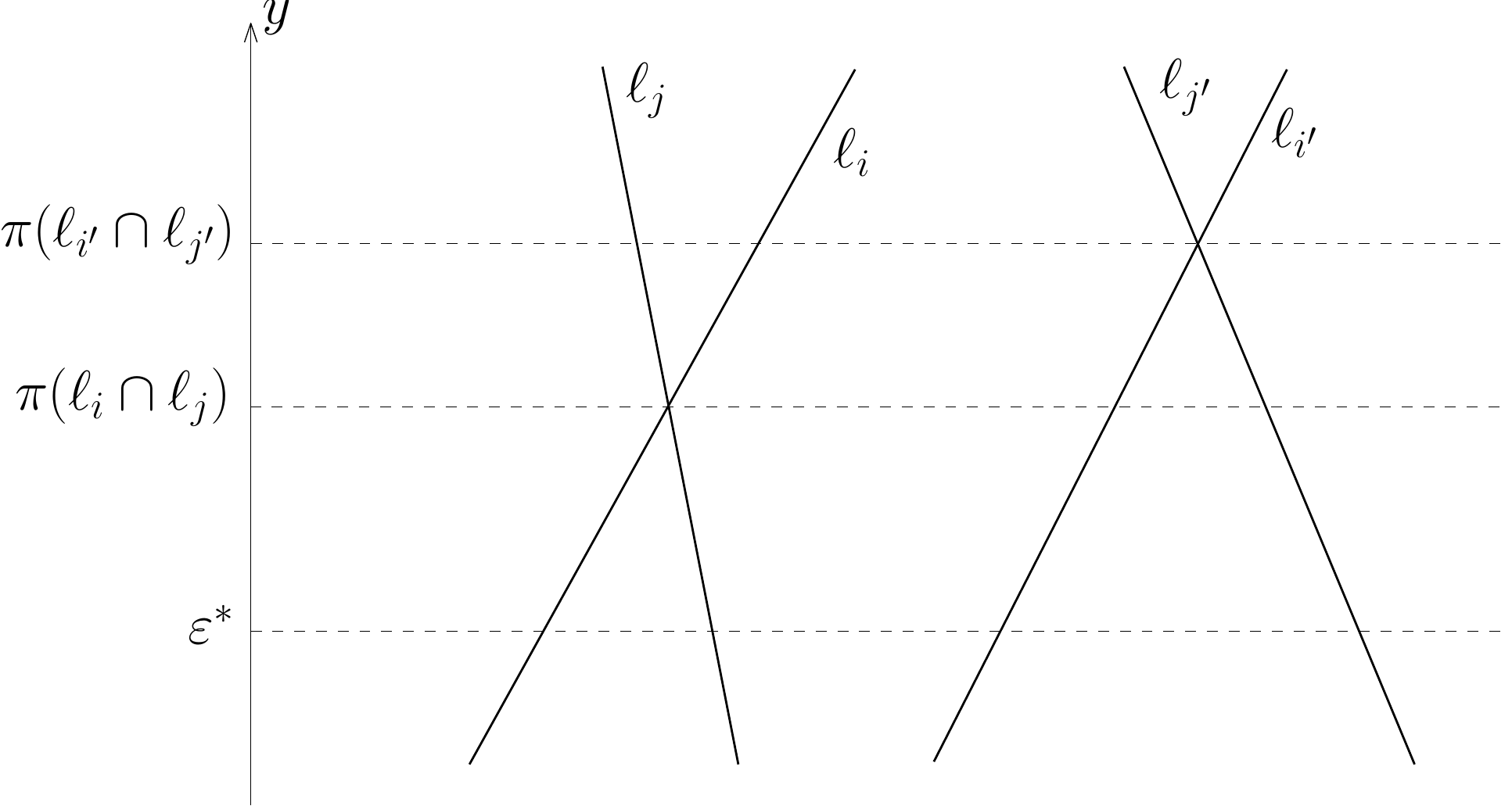}}
\end{center}
\caption{A case where $(i,j) \preceq (i',j')$ and $\ell_i(\eps^*) \leq \ell_j(\eps^*)$.\label{fig:order}}
\end{figure}
From the second condition, it holds that $\eps^* \leq \tilde\pi(\ell_i,\ell_j)$,
then the first condition yields $\eps^* \leq \tilde\pi(\ell_{i'},\ell_{j'})$;
it follows that $\ell_{i'}(\eps^*) \leq \ell_{j'}(\eps^*)$.
Moreover, the order $\preceq$ can obviously be computed in $O(1)$ time.

Hence Theorem~\ref{th:cole} allows to compute a permutation
$\sigma$ of $\{1,\ldots,m\}$ such that 
\[\ell_{\sigma(1)}(\eps^*) \leq \ell_{\sigma(2)}(\eps^*) \leq \ldots \leq \ell_{\sigma(m)}(\eps^*)\]
in time $O(n \log n)$.
By Lemma~\ref{lem:intersection}, it holds that
$\eps^* \in \{ \pi(\ell_{\sigma(i)} \cap \ell_{\sigma(i+1)}) \mid 0 < i < m\}$.
After sorting this set, we perform a binary search using the linear-time decision 
algorithm, and thus we compute $\eps^*$ in $O(n \log n)$ time.
At last, we run the decision algorithm on $\eps^*$, which 
gives an optimal $k$-step function for $P$.
\end{proof}

\section{Concluding remarks}

When the input points are given in unsorted order,
our algorithm is optimal by reduction from sorting~\cite{FournierV11}.
So an intriguing question is whether there exists an $o(n \log n)$-time
algorithm when the input points are given in sorted order.
For instance, in the unweighted case, a linear-time algorithm
exists if the points are sorted according to their $x$-coordinates~\cite{FournierV11}.

Our algorithm is mainly of theoretical interest,
as Cole's parametric searching technique relies on a sorting network
with $O(\log n)$ depth; all known constructions for such networks 
involve large constants.
So it would be interesting to have a practical
$O(n \log n)$-time deterministic algorithm.


{\small
\bibliography{wstep}
}

\end{document}